\documentclass{article}[12pt]
\usepackage{graphicx,epsfig,color}
\usepackage{enumerate}
\usepackage{amsfonts}

\newcounter{savenumi}

\newtheorem{theoremfoo}{Theorem}
\newenvironment{theorem}{\pagebreak[1]\begin{theoremfoo}}{\end{theoremfoo}}

\newtheorem{propositionfoo}[theoremfoo]{Proposition}
\newenvironment{proposition}{\pagebreak[1]\begin{propositionfoo}}{\end{propositionfoo}}
\newtheorem{lemmafoo}[theoremfoo]{Lemma}
\newenvironment{lemma}{\pagebreak[1]\begin{lemmafoo}}{\end{lemmafoo}}
\newtheorem{conjecturefoo}[theoremfoo]{Conjecture}

\newtheorem{corollaryfoo}[theoremfoo]{Corollary}
\newenvironment{corollary}{\pagebreak[1]\begin{corollaryfoo}}{\end{corollaryfoo}}
\newtheorem{exercisefoo}{Exercise}

\newtheorem{openfoo}[theoremfoo]{Question}

\newtheorem{nttn}[theoremfoo]{Notation}

\newtheorem{dfntn}[theoremfoo]{Definition}
\newenvironment{definition}{\pagebreak[1]\begin{dfntn}\rm}{\end{dfntn}}

\newenvironment{proof}
    {\pagebreak[1]{\narrower\noindent {\bf Proof:\quad\nopagebreak}}}{\QED}

\newcommand{\ceiling}[1]{\left\lceil#1\right\rceil}

\newcommand{\poly}{{\rm poly}}
\def\nre.{$n$\/-r.e.}

\newtheorem{factfoo}[theoremfoo]{Fact}

\newcommand{\squeeze}{
\textwidth 6in \textheight 8.8in \oddsidemargin 0.2in \topmargin
-0.4in }

\newtheorem{propertyfoo}[theoremfoo]{Property}

\makeatletter

\def\@makechapterhead#1{ \vspace*{50pt} { \parindent 0pt \raggedright
 \ifnum \c@secnumdepth >\m@ne \huge\bf \@chapapp{} \thechapter. \par
 \vskip 20pt \fi \Huge \bf #1\par
 \nobreak \vskip 40pt } }

\def\@sect#1#2#3#4#5#6[#7]#8{\ifnum #2>\c@secnumdepth
     \def\@svsec{}\else
     \refstepcounter{#1}\edef\@svsec{\csname the#1\endcsname.\hskip 1em }\fi
     \@tempskipa #5\relax
      \ifdim \@tempskipa>\z@
        \begingroup #6\relax
          \@hangfrom{\hskip #3\relax\@svsec}{\interlinepenalty \@M #8\par}
        \endgroup
       \csname #1mark\endcsname{#7}\addcontentsline
         {toc}{#1}{\ifnum #2>\c@secnumdepth \else
                      \protect\numberline{\csname the#1\endcsname}\fi
                    #7}\else
        \def\@svsechd{#6\hskip #3\@svsec #8\csname #1mark\endcsname
                      {#7}\addcontentsline
                           {toc}{#1}{\ifnum #2>\c@secnumdepth \else
                             \protect\numberline{\csname the#1\endcsname}\fi
                       #7}}\fi
     \@xsect{#5}}

\def\@begintheorem#1#2{\it \trivlist \item[\hskip \labelsep{\bf #1\ #2.}]}

\def\@opargbegintheorem#1#2#3{\it \trivlist
      \item[\hskip \labelsep{\bf #1\ #2\ (#3).}]}

\makeatother

\newif\ifshortconferences
\shortconferencesfalse
\newif\ifmediumconferences
\mediumconferencesfalse

\def\ending#1{{\count1=#1\relax
\count2=\count1 \divide\count2 by 100 \multiply\count2 by 100
\advance\count1 by -\count2
\ifnum\count1=11
th%
\else \ifnum\count1=12
th%
\else \ifnum\count1=13
th%
\else \count2=\count1 \divide\count1 by 10 \multiply\count1 by 10
\advance\count2 by -\count1 \ifnum\count2=1
st%
\else \ifnum\count2=2
nd%
\else \ifnum\count2=3
rd%
\else th%
\fi\fi\fi\fi\fi\fi }}

\def\Proceedingsofthe{\ifshortconferences Proc.\else\ifmediumconferences Proc.\else Proceedings of the\fi\fi}

\newcounter{confnum}

\def\conf#1#2{%
\setcounter{confnum}{#2}%
\addtocounter{confnum}{-\csname #1zero\endcsname}%
\ifnum\value{confnum}=1%
\expandafter\ifx\csname #1One\endcsname\relax%
\Proceedingsofthe\ \arabic{confnum}\ending{\value{confnum}}\ \csname #1name\endcsname%
\else \csname #1One\endcsname\fi%
\else%
\Proceedingsofthe\ \arabic{confnum}\ending{\value{confnum}}\ \csname
#1name\endcsname\fi}

\def\qsym{\vrule width0.7ex height0.9em depth0ex}

\newif\ifqed\qedtrue

\def\noqed{\global\qedfalse}

\def\qed{\ifqed{\penalty1000\unskip\nobreak\hfil\penalty50
\hskip2em\hbox{}\nobreak\hfil\qsym
\parfillskip=0pt \finalhyphendemerits=0\par\medskip}\fi\global\qedtrue}

\makeatletter
\def\eqnqed{\noqed
    \def\@tempa{equation}
    \ifx\@tempa\@currenvir\def\@eqnnum{\qsym}%
    \addtocounter{equation}{-1}\else%
    \def\@@eqncr{\let\@tempa\relax
    \ifcase\@eqcnt \def\@tempa{& & &}\or \def\@tempa{& &}%
      \else \def\@tempa{&}\fi
     \@tempa {\def\@eqnnum{{\qsym}}\@eqnnum}
     \global\@eqnswtrue\global\@eqcnt\z@\cr}\fi}

\def\eqnlabel#1#2{\if@filesw {\let\thepage\relax%
   \def\protect{\noexpand\noexpand\noexpand}%
   \edef\@tempa{\write\@auxout{\string
      \newlabel{#2}{{{#1}}{\thepage}}}}%
   \expandafter}\@tempa%
   \if@nobreak \ifvmode\nobreak\fi\fi\fi%
    \def\@tempa{equation}
    \ifx\@tempa\@currenvir\def\theequation{{#1}}%
    \addtocounter{equation}{-1}\else%
    \def\@@eqncr{\let\@tempa\relax
    \ifcase\@eqcnt \def\@tempa{& & &}\or \def\@tempa{& &}%
      \else \def\@tempa{&}\fi
     \@tempa {\def\@eqnnum{{#1}}\@eqnnum}
     \global\@eqnswtrue\global\@eqcnt\z@\cr}\fi}

\makeatother

\def\QED{\qed}

\makeatother

\squeeze

\newcommand{\sampleCount}{ s}
\newcommand{\dist}{{\rm dist}}

\newcommand{\bigO}{O}

\newcommand{\gammabound}{{\gamma\over 6}} 

\begin{document}

\date{}

\title{Speeding Up Constrained $k$-Means Through 2-Means }

\author{
Qilong Feng$^1$ and Bin Fu$^2$
 \\ \\
$^1$School of Information Science and Engineering\\
 Central South 
University, P.R. China\\ csufeng@csu.edu.cn\\
\\
$^2$Department of Computer Science\\
 University of Texas Rio Grande Valley,  USA\\ bin.fu@utrgv.edu\\
} \maketitle


\begin{abstract}
For the constrained 2-means problem, we present a $\bigO\left(dn+d({1\over\epsilon})^{\bigO({1\over \epsilon})}\log n\right)$ time algorithm. It  generates a collection $U$ of approximate center pairs $(c_1, c_2)$ such that one of pairs in $U$ can induce a $(1+\epsilon)$-approximation for the problem. The existing approximation scheme for the
constrained 2-means problem takes
$\bigO(({1\over\epsilon})^{\bigO({1\over \epsilon})}dn)$ time, and
the existing approximation scheme for the constrained $k$-means
problem takes $\bigO(({k\over\epsilon})^{\bigO({k\over
\epsilon})}dn)$ time. Using the method developed in this paper, we
point out that every existing approximating scheme for the constrained
$k$-means so far with time $C(k, n, d, \epsilon)$ can be transformed
to a new approximation scheme with time complexity ${C(k, n, d,
\epsilon)/ k^{\Omega({1\over\epsilon})}}$.
\end{abstract}


\section{Introduction}

The $k$-means problems is to partition a set $P$ of points in
$d$-dimensional space $\mathbb{R}^d$ into $k$ subsets $P_1,\cdots, P_k$ such
that  $\sum_{i=1}^k\sum_{p\in P_i}||p-c_i||^2$ is minimized, where
$c_i$ is the center of $P_i$, and $||p-q||$ is the distance between
two points of $p$ and $q$. The $k$-means problem is one of the
classical NP-hard problems in the field of computer science, and has
broad applications as well as theoretical importance. The $k$-means
problem is NP-hard even for the case $k=2$~\cite{AloiseDHP09}. The
classical $k$-means problem and $k$-median problem have received a lot of attentions in the last decades~\cite{Matousek00,BadoiuHI02,Vega,Har-PeledM04,KumarSS10,AckermannBS10,Chen06,Jaiswal0S14,FeldmanMS07,OstrovskyRSS12}.

Inaba, Katoh, and Imai~\cite{InabaKI94} showed that $k$-means problem has an
exact algorithm~\cite{InabaKI94} with running time $\bigO(n^{kd+1})$. For the $k$-means problem, Arthur and Vassilvitskii \cite{Arthur} gave a $\Theta (\log k$)-approximation algorithm. A $(1+\epsilon)$-approximation
scheme was derived by de la Vega et al.~\cite{Vega} with time $\bigO(e^{{k^3/\epsilon^8}\ln (k/\epsilon)\ln
    k}\cdot n\log^kn)$. Kumar, Sabharwal,
and Sen~\cite{KumarSS10} presented a
$(1+\epsilon)$-approximation algorithm for the $k$-means problem with running time $\bigO(2^{(k/\epsilon)^{\bigO(1)}}nd)$.
Ostrovsky et al. ~\cite{OstrovskyRSS12}
developed a $(1+\epsilon)$-approximation for the $k$-means problem under the
separation condition with running time $\bigO(2^{\bigO(k/\epsilon)}nd)$.
Feldman,
Monemizadeh, and Sohler~\cite{FeldmanMS07} gave a $(1+\epsilon)$-approximation scheme for the $k$-means problem using
corset with running time $\bigO(knd+d\cdot
\poly(k/\epsilon)+2^{\tilde{\bigO}{(k/\epsilon)}})$. Jaiswal, Kumar, and Yadav~\cite{Jaiswal0S14} presented
a $(1+\epsilon)$-approximation algorithm for the $k$-means problem
using $D^2$-sampling method with running time $\bigO(2^{\tilde{\bigO}(k^2/\epsilon)}nd)$. Jaiswal, Kumar, and Yadav~\cite{Jaiswal} gave a $(1+\epsilon)$-approximation algorithm with running time $\bigO(2^{\tilde{\bigO}(k/\epsilon)}nd)$. Kanungo et al. \cite{Kanungo} presented a $(9+\epsilon)$-approximation algorithm for the problem in polynomial time by applying local search.  Ahmadian et al. \cite{Ahmadian} gave a $(6.375+\epsilon)$-approximation algorithm for the $k$-means problem in Euclidean space.
For fixed $d$ and arbitrary $k$, Friggstad, Rezapour, and Salavatipour \cite{ls1} and Cohen-Addad, Klein, and Mathieu \cite{ls2} proved that the local search algorithm yields a PTAS for the problem, which runs in $(nk)^{(1/\epsilon)^{O(d)}}$ time. Cohen-Addad \cite{Cohenaddad2017A} further showed that the running time can be improved to $nk(\log n)^{(d/\epsilon)^{O(d)}}$.



The input data of the $k$-means problem always satisfies local properties. However, for many applications, each cluster of the input data may satisfy some additional constraints.  It seems that the
constrained $k$-means problem has different structure from the
classical $k$-means problem, which lets each point go to the cluster
with nearest center. The constrained $k$-means problems have been paid lots of attention in the literature, such as the chromatic clustering problem \cite{arkin2015bichromatic,chromatic}, the $r$-capacity clustering problem \cite{capa}, $r$-gather clustering \cite{r}, fault tolerant clustering \cite{FAULT}, uncertain data clustering \cite{uncertain}, semi-supervised clustering \cite{back,valls2009using}, and $l$-diversity clustering \cite{diver}. As given in  Ding and Xu \cite{cmeans1}, all $k$-means problems with constraint conditions can be defined as follows.

\begin{definition}[Constrained $k$-means problem]
    Given a point set $P\subseteq\mathbb{R}^d$, a list of constraints $\mathbb{L}$, and a positive integer $k$, the constrained $k$-means problem is to partition $P$ into $k$ clusters $\mathbb{P}=\{P_1\ldots P_k\}$ such that all the constraints in $\mathbb{L}$ are satisfied and $\sum_{P_i\in \mathbb{P}}\sum_{x\in P_i}||x-c(P_i)||^2$ is minimized, where $c(P_i)=\frac{1}{|P_i|}\sum_{x\in P_i}x$ denotes the centroid of $P_i$.
\end{definition}

Recent years, there are some progress for the constrained $k$-means
problem. The first polynomial time approximation scheme with running
time $\bigO(2^{\poly(k/\epsilon)}(\log n)^knd)$ for the constrained
$k$-means problem was shown by Ding and Xu~\cite{cmeans1}, and a collection of size $O(2^{poly (k/\epsilon)}(\log n)^{k+1})$ of candidate approximate centers can be obtained. The
existing fastest approximation schemes for the constrained $k$-means problem
takes
$\bigO(nd\cdot 2^{\tilde{\bigO}(k/\epsilon)})$ time~\cite{BhattacharyaJaiswalKumar17,bhattacharya2018faster,FengHuHuangWang17}, which was first
derived by Bhattacharya, Jaiswai, and
Kumar~\cite{BhattacharyaJaiswalKumar17,bhattacharya2018faster}. Their algorithm gives a collection of size $O(2^{\tilde{O}(k/\epsilon)})$ of candidate approximate centers. Feng et al.~\cite{FengHuHuangWang17} analyzed the complexity
of~\cite{BhattacharyaJaiswalKumar17,bhattacharya2018faster} and gave an algorithm with running time $O((\frac{1891ek}{\epsilon^2})^{8k/\epsilon}nd)$, which outputs a collection of size $O((\frac{1891ek}{\epsilon^2})^{8k/\epsilon}n)$ of candidate approximate centers.

It is known that 2-means problem is the smallest version of the $k$-means problem, and remains being NP-hard. Obviously, all the approximation algorithms of the $k$-means problem can be directly applied to get approximation algorithms for the 2-means problem. However, not all the approximation algorithms for 2-means problem can be generalized to solve the $k$-means problem. The understanding of the characteristics of the 2-means problem will give new insight to the $k$-means problem. Meanwhile,  getting two clusters of the input data is useful in many interesting applications, such as the ``good" and ``bad" clusters of input data, the ``normal" and ``abnormal" clusters of input data, etc.

For the 2-means problem, Inaba, Katoh, and Imai~\cite{InabaKI94} presented an
$(1+\epsilon)$-approximation scheme for $2$-means with running time
$\bigO(n({1\over \epsilon})^d)$.  Matou{\v{s}}ek~\cite{Matou} gave a deterministic $(1+\epsilon)$-approximation algorithm with running time $O(n$log$ n(1/\epsilon)^{d-1})$. Sabharwal and Sen ~\cite{sabharwal2005} presented a $(1+\epsilon)$-approximation algorithm with linear running time $O((1/\epsilon)^{O(1/\epsilon)}(d/\epsilon)^dn)$. Kumar, Sabharwal, and Sen~\cite{kumar2004} gave a randomized approximation algorithm with running time $O(2^{(1/\epsilon)^{O(1)}}dn)$.

This paper develops a new technology to deal with the constrained 2-means problem. It is based on how balance between the sizes of clusters in the constrained $2$-means problem.  This brings an algorithm  with running time $\bigO(dn+d({1\over\epsilon})^{\bigO({1\over \epsilon})}\log n)$. Our algorithm outputs a collection of size $O(({1\over\epsilon})^{\bigO({1\over \epsilon})}\log n)$ of candidate approximate centers, in which one of them induces a $(1+\epsilon)$-approximation for the constrained $2$-means problem. The technology shows a faster way to obtain first two approximate centers when applied to the constrained $k$-means, and can speed up the existing approximation schemes for constrained $k$-means with
$k$ greater than 2.
Using this method developed in this paper, we point out every
existing PTAS for the constrained $k$-means so far with time $C(k,
n, d, \epsilon)$ can be transformed to a new PTAS with time
complexity ${C(k, n, d, \epsilon)/ k^{\Omega({1\over\epsilon})}}$.
Therefore, we provide a unified
approach to speed up the existing approximation scheme for the
constrained $k$-means problem.


This papers is organized with a few sections. In Section~\ref{prelim-sect}, we give some basic notations. In
section~\ref{overview-sect}, we give an overview of the new
algorithm for the constrained $2$-means problem. In
section~\ref{2-means-sect}, we give a much faster approximation
scheme for the constrained $2$-means problem. In
section~\ref{k-mean-sect}, we apply the method to the general
constrained $k$-means problem, and show faster approximation
schemes.

\section{Preliminaries}\label{prelim-sect}

This section gives some notations that are used in the algorithm design.

\begin{definition} Let $c$ be a real number in $[1,+\infty)$. Let $P$ be a set of points in $\mathbb{R}^d$.
    \begin{itemize}
        \item
        A partition $Q_0,Q_1$ of $P$ is $c$-balanced if $|Q_i|\le
        c|Q_{1-i}|$ for $i=0,1$.
        \item
        A $c$-balanced $k$-means problems is to partition $P$ into $P_1,\cdots, P_k$ such that $|P_i|\le c|P_j|$ for all $1\le i,j\le k$.
    \end{itemize}
\end{definition}

\begin{definition}\label{f2-def}
    Let $S$ be a set of points in $\mathbb{R}^d$, and $q\in
    \mathbb{R}^d$.
    \begin{itemize}
        \item
        Define $f_2(q, S)=\sum_{p\in S} ||p-q||^2$.
        \item
        Define $c(S)={1\over |S|}\sum_{p\in S}p$.
    \end{itemize}
\end{definition}

\begin{definition}\label{OPT-def} Let $P$ be a set of points in $\mathbb{R}^d$, and
    $P_1,\cdots, P_k$ be a partition of $P$.
    \begin{enumerate}[1.]
        \item
        Define $m_j=c(P_j)$.
        \item
        Define $\beta_j={|P_j|\over |P|}$.
        \item\label{sigma-j-def}
        Define $\sigma_j=\sqrt{{f_2(m_j, P_j)\over |P_j|}}$.
        \item
        Define $OPT_k(P)=\sum_{j=1}^k\sum_{p\in P_j}||p-c(P_j)||^2$.
        \item\label{sigma-opt2-def}
        Define $\sigma_{opt}=\sqrt{{OPT_k(P)\over
                |P|}}=\sqrt{\sum_{i=1}^k\beta_i\sigma_i^2}$.
    \end{enumerate}
\end{definition}

Chernoff Bound
(see~\cite{MotwaniRaghavan00}) is used in the approximation
algorithm when our main result is applied in some concrete model.

\begin{theorem}\label{chernoff3-theorem}
Let $X_1,\ldots , X_{\sampleCount}$ be $\sampleCount$ independent
random $0$-$1$ variables, where $X_i$ takes $1$ with probability at
least $p$ for $i=1,\ldots , \sampleCount$. Let
$X=\sum_{i=1}^{\sampleCount} X_i$. Then for any
$\delta>0$,
 $\Pr(X<(1-\delta)p\sampleCount)<e^{-{1\over 2}\delta^2 p\sampleCount}$.
\end{theorem}

The union bound is
expressed by the inequality
\begin{eqnarray}\Pr(E_1\cup E_2 \ldots \cup E_{m})\le
\Pr(E_1)+\Pr(E_2)+\ldots+\Pr(E_{m}),\label{prob-Add-Ineqn}
\end{eqnarray}
where $E_1,E_2,\ldots, E_{m}$ are
$m$ events that may not be independent.
We will use the famous Stirling formula
\begin{eqnarray}
n!\approx \sqrt{2\pi n}\cdot {n^n\over e^n}.
\end{eqnarray}





For two points $p=(x_1,x_2,\cdots, x_d)$ and $q=(y_1, y_2,\cdots, y_d)$ in $\mathbb{R}^d$, both $\dist(p,q)$ and $||p-q||$ represent their Euclidean distance $\sqrt{\sum_{i=1}^d (x_i-y_i)^2}$.  For a finite set $S$, $|S|$ is the number
of elements in it.

\begin{lemma}\cite{KumarSS10}\label{KumarSS10-eqn} For a set $P\subseteq \mathbb{R}^d$ of points, and any point
$x\in \mathbb{R}^d$, $f_2(x, P)=f_2(c(P),P)+|P|||c(P)-x||^2$.
\end{lemma}

\begin{lemma}\cite{InabaKI94}\label{geometric-random-lemma}
Let $S$ be a set of points in $\mathbb{R}^d$. Assume that $T$ is a set of
points obtained by sampling points from $S$ uniformly and independently. Then
for any $\delta>0$, $||c(T)-c(S)||^2\le {1\over \delta |T|}\sigma^2$
with probability at least $1-\delta$, where $\sigma^2={1\over
|S|}\sum_{q\in S}||q-c(S)||^2$.
\end{lemma}

\begin{lemma}\cite{cmeans1}\label{large-portion-lemma}
Let $Q$ be a set of points in $\mathbb{R}^d$, and $Q_1$ be an arbitrary subset
of $Q$ with $\alpha |Q|$ points for some $0<\alpha \le 1$. Then
$||c(Q)-c(Q_1)||\le \sqrt{1-\alpha\over\alpha}\sigma$, where
$\sigma^2={1\over |Q|}\sum_{q\in Q}||q-c(Q)||^2$.
\end{lemma}

\begin{lemma}~\cite{BhattacharyaJaiswalKumar17,bhattacharya2018faster}\label{norm-square-lemmma}
 For any three points $x, y,z\in \mathbb{R}^d$, we have $||x-z||^2\le
 2||x-y||^2+2||y-z||^2$.
\end{lemma}

\section{Overview of Our Method}\label{overview-sect}

In order to develop a faster algorithm for the constrained $2$-Means
problem, we assume that the input set $P$ has two clusters $P_1$ and $P_2$.
We will try to find a subset $H_1$ and $H_2$ of size $M$ from $P_1$
and $P_2$, respectively, where $M$ is an integer to be large enough
to derive an approximate center by Lemma~\ref{large-portion-lemma}.
We consider two different cases. The first case is that the two
clusters $P_1$ and $P_2$ with $|P_1|\ge |P_2|$ have a balanced sizes
of points ($|P_1|\ge |P_2|=\epsilon^{O(1)}|P|)$. We get a set
$N_a$ of random samples, and another set $N_b$ of random
samples from $P$. An approximate center $c_1$ for the cluster $P_1$
will be generated via one of the subsets of size $M$ from  $N_a$. An
approximate center $c_2$ for $P_2$ will be generated via one of the subsets of size $M$ from
$N_b$. The two parameters  $N_a$ and  $N_b$ are selected based on the balanced condition between the sizes of $P_1$ and $P_2$.

We discuss the case that $|P_1|$ is much larger than $|P_2|$. We
generate a subset $V_a\subseteq P_1$  with $|V_a|=M$  that will be
used to generate an approximate center $c_1=c(V_a)$ for $P_1$. The
set $V_a$ can be obtained via $M$ random samples from $P$ since
$|P_1|$ is much larger than $|P_2|$.
It also has
two cases to find another approximate center $c_2$ for $P_2$.
The
first case is that almost all points of $P_1$ is close to $c_1$. In
this case, we just let $c_2$ be the same as $c_1$, which is based on
Lemma~\ref{large-portion-lemma}. The second case is that there are
enough points of $P_2$ to be far from $c_1$. This transforms the
problem into finding the second approximate center $c_2$ for the
second cluster $P_2$ assuming the approximate center $c_1$ is good
enough for $P_1$.

Phase $0$ of the algorithm lets $Q_0$ be equal to $P$. Phase $i+1$ extracts the set of half elements $Q_{i+1}$ from $Q_i$ with larger
distances to $c_1$ than the rest half.
It will have $\bigO(\log n)$ phases
to search $c_2$. The next phase will shrink the search area by a
constant factor. This method was used in the existing algorithms. As
we only have one approximate center for $P_1$, it saves the amount
of  time by a factor to find the first approximate center. This
makes our approximation algorithm run in $\bigO(dn+({1\over\epsilon})^{\bigO({1\over \epsilon})}\log n)$ time for the
constrained $2$-means problem.

\section{Approximation Algorithm for Constrained  2-means}\label{2-means-sect}

In this section, an approximation scheme will be presented for the constrained
$2$-means problem. The methods used in this section will be applied
to the general constrained $k$-means problem in
Section~\ref{k-mean-sect}. We define some parameters before
describing the algorithm for the constrained $2$-means  problem.


\subsection{Setting Parameters}\label{parameters-sect}

Assume that real parameter $\epsilon\in (0,1)$ is used to control the  approximation ratio, and real parameter $\gamma\in (0,1)$ is used to control failure
probability of the randomized algorithm.
We define some constants for our algorithm and its
analysis.
All the parameters that are set up through (\ref{gamma-eqn}) to (\ref{varsigma-delta1-ineqn}) in this section are positive real constants.

\begin{eqnarray}
\gamma&=&{1\over 3},\label{gamma-eqn}\\
\delta=\delta_6&=&{1\over 10},\label{delta6-eqn}\\
d_1&=&\delta_6,\\
d_2&=&4,\label{d2-eqn}.
\end{eqnarray}

We select $\delta_2\in (0,{1\over 2}]$ to satisfy
inequality~(\ref{delta2-d2-ineqn}).
\begin{eqnarray}
(1+2\delta_2)\cdot d_2\le d_2 +\delta. \label{delta2-d2-ineqn}
\end{eqnarray}

\begin{eqnarray}
d_4&=&{(10+
    18\delta_6)\over \delta_2^2}\cdot ({1\over \gamma}+\log {1\over \gamma}), \label{c4-eqn}\\
\gamma^*&=&{1\over 2}-\delta_6,\label{gamma4-eqn}\\
\gamma_{5,b}&=&{\gamma\over 12},\label{gamma5-b-eqn}\\
\eta&=&{\delta_6\over 2},\label{eta-eqn}\\
\alpha_1&=&5,\label{alpha1-eqn}\\
\alpha_2&=&2\alpha_1\alpha_5+\delta_6,\label{alpha2-eqn}\\
\alpha_6&=&{\gamma^*\cdot d_4\over 12}, \label{alpha6-eqn}\\
 \alpha_5&=&{4\over
1-\delta_6-{4\over \alpha_6}}, {\rm \ and}\label{alpha5-eqn}\\
M&=&{d_4\over \epsilon}.\label{M-eqn}
\end{eqnarray}

We select $\varsigma$ and $\delta_1$ in $(0,{1\over 2}]$ to satisfy
inequality (\ref{varsigma-delta1-ineqn}).
\begin{eqnarray}
(1+\varsigma)(1+2\delta_1)\alpha_2&\le& \alpha_2+\delta/2.
\label{varsigma-delta1-ineqn}
\end{eqnarray}

\begin{lemma}\label{parameters-lemma} The parameters satisfy the following conditions (\ref{constant-start-conditions}) to (\ref{constant-end-conditions}):
\begin{eqnarray}
\alpha_6&\ge&4,\label{constant-start-conditions}\\
e^{{-\delta_2^2d_2\over 4}M}&\le& \gammabound, \label{exp-M-ineqn}\\
e^{-aM}&\le& {\gamma\over 12}
\ \ {\rm for \ any \ positive\ real\ }a {\rm \ and \ all \ } \epsilon\le  {ad_4\over \ln {12\over \gamma}},\label{exp-M2-ineqn}\\
{4\over \alpha_6}+{4\over
\alpha_5}&=&1-\delta_6.\label{constant-end-conditions}
\end{eqnarray}
\end{lemma}

\begin{proof}
Inequality (\ref{constant-start-conditions}): By equations 
(\ref{alpha6-eqn}), (\ref{c4-eqn}) and (\ref{gamma4-eqn}), we have inequalities:
\begin{eqnarray}
\alpha_6&=&{1\over 12}\cdot \gamma^*\cdot d_4={1\over 12}\cdot ({1\over 2}-\delta_6)\cdot (10\cdot 4\cdot 3)\ge 4.
\end{eqnarray}

Inequality (\ref{exp-M-ineqn}):
Let $z={\delta_2^2d_2\over 4}M$. We have the inequalities:
\begin{eqnarray}
z&=&{\delta_2^2d_2\over 4}\cdot M\\
&\ge&{\delta_2^2d_2\over 4}\cdot {d_4\over \epsilon}\\
&\ge&{d_2\over 4}\cdot 8({1\over \gamma}+\log {1\over \gamma})\\
&\ge&2d_2({1\over \gamma}+\log {1\over \gamma}).
\end{eqnarray}
Thus, $e^{-z}\le \gammabound$.

Inequality (\ref{exp-M2-ineqn}):
By equation (\ref{M-eqn}) we have $e^{-aM}\le {\gamma\over 12} $ when  $\epsilon\le {ad_4\over \ln {12\over \gamma}}$.

Equation (\ref{constant-end-conditions}): It follows from equation (\ref{alpha5-eqn}).
\end{proof}

\subsection{Algorithm Description}

In this section, an approximation algorithm for the constrained 2-means problem
is given. It outputs a collection of centers, and one of them brings a $(1+\epsilon)$-approximation for the constrained 2-means problem.

\vskip 20pt {\bf Algorithm} $2$-Means$(P, \epsilon)$

Input: $P$ is a set of points in $\mathbb{R}^d$, and real parameter
$\epsilon\in (0,1)$ to control accuracy of approximation.

Output: A collection $U$ of two centers $(c_1,c_2)$.

\begin{enumerate}[1.]



\item
\qquad Let $U=\emptyset$;

\item \qquad Let $j=0$;

\item\label{M-setting-line}
\qquad Let $M$ be defined as that in equation (\ref{M-eqn});


\item

\qquad Let $N_a=d_2M$;

\item\label{N1-line}
\qquad Let $N_b={M\over (1-\delta_2)\cdot \epsilon^{1+d_1}}$;


\item\label{N3-line}
\qquad Let $N_2=M\cdot (1+\varsigma)\cdot {\alpha_2\over
\epsilon^2}\cdot {1\over 1-\delta_1}$;

\item\label{phase1-start-line}
\qquad Select a set $S_a$ of $N_a$ random samples from $P$;

\item
\qquad Select a set $S_b$ of $N_b$ random samples from $P$;

\item
\qquad For every two subsets $H_1$ of $S_a$ and $H_2$ of
$S_b$  of size $M$,

\item
\qquad \{

\item
\qquad\qquad Compute the centroid $c(H_1)$ of $H_1$, and
$c(H_2)$ of $H_2$;


\item
\qquad\qquad Add $(c(H_1), c(H_2))$ to $U$;

\item\label{phase1-end-line}
\qquad \}

\item\label{phase2-start-line}
\qquad Select a set $V_a$ of $M$ random samples from $P$;

\item
\qquad Compute the centroid $c_1=c(V_a)$ of $V_a$;


\item
\qquad Let $Q_0=P$;

\item\label{phase2-loop-start-line}
\qquad Repeat

\item
\qquad \qquad Select a set $V_b $ of $N_2$ random samples from $Q_j$;

\item\label{phase2-loop2-start-line}
\qquad\qquad For each size $M$ subset $H'$ of $V_b \cup \{M$ copies of
$c(V_b )\}$

\item
\qquad\qquad \{

\item\label{H'-center-line}
\qquad\qquad\qquad Compute the centroid $c(H')$ of $H'$;

\item
\qquad\qquad\qquad  Add $(c_1, c(H'))$ to $U$;

\item\label{phase2-loop2-end-line}
\qquad\qquad \}

\item\label{find-i-th-largest-line}
\qquad\qquad Let $d_j$ be the ${|Q_j|\over 1+\varsigma}$-th largest of
$\{\dist(p, c_1): p\in Q_j\}$;

\item
\qquad\qquad Let $Q_{j+1}$ contain all of the points $q$
in $Q_j$ with $\dist(q,c_1)\ge d_j$;

\item
\qquad\qquad Let $j=j+1$;

\item\label{phase2-end-line}
\qquad Until $Q_j$ is empty;

\item
\qquad Output $U$;

\end{enumerate}

{\bf End of Algorithm}

\begin{definition}\label{rj-Pj-mj-def} Let $c_1$ be the  approximate center of
    $P_1$ via the algorithm.
    \begin{enumerate}[1.]
        \item\label{rj-def}
        Define $r_2=\sqrt{\epsilon\over \alpha_5\beta_2}\sigma_{opt}$.
        \item
        Define $B_2=\{p\in P:||p-c_{1}||\le r_2\}$.

        \item
        Define $P_2^{out}=P_2-B_2$.

        \item
        Define $P_2^{in}=P_2\cap B_2$.

        \item
        Let $m_j$ be the center of $P_j$ for $j=1,2$.

        \item
        Let $m_j^{in}$ be the center for $P_j^{in}$ for $j=1,2$

            \item
        For each $p\in {P_2^{in}}$, let $\tilde{p}=c_1$.

        \item
        Let $\tilde{P_2^{in}}$ be the multiset with $|P_2^{in}|$ number
        of $c_1$. It transforms every element of $P_2^{in}$ to $c_1$.


        \item
        Let $\tilde{P_2}=\tilde{P_2^{in}}\cup P_2^{out}$.

        \item
        Let $\tilde{m_j}$ be the center of $\tilde{P_j^{in}}\cup P_j^{out}$  for $j=1,2$.
    \end{enumerate}
\end{definition}

\begin{lemma}\label{ex-lemma}
Let $x$ be a real number in $[0,1]$ and $y$ be positive real number with $1\le y$. Then we have
$1-xy\le (1-x)^y$,
\end{lemma}

\begin{proof}
By Taylor formula, we have $(1-x)^y=1-xy+{y\cdot (y-1)\over 2}\xi^2$
for some $\xi\in [0,x]$. Thus, we have $(1-x)^y\ge 1- y x$.
\end{proof}

\begin{lemma}\label{early-sampling-lemma} The algorithm $2$-Means(.) has the following
properties:
\begin{enumerate}[1.]
\item\label{early-case1}
With probability at least $1-\gamma_1$, at least ${1\over 2}(1-\delta_2)d_2M$
random points are from $P_1$ in $S_a$, where
$\gamma_1=e^{{-\delta_2^2d_2\over 4}M}\le \gammabound$.
\item\label{early-case2}
If the two clusters $P_1$ and $P_2$ satisfy $\epsilon^{1+d_1}|P|\le
|P_2|\le |P_1|$, then with probability at least $1-\gamma_2$,  at
least $M$ random points are from $P_2$ in $S_b$, where
$\gamma_2=e^{{-\delta_2^2\over 2}\cdot {M\over 1-\delta_2}}\le
\gammabound$.
\item\label{phase1-complexity-statm}
Line \ref{phase1-start-line} to line~\ref{phase1-end-line} of the
algorithm $2$-Means(.) generate at most ${N_a\choose M}\cdot
{N_b\choose M}$ pairs of centers.
\item\label{inbalance-statm}
If the clusters $P_1$ and $P_2$ satisfy $|P_2|<\epsilon^{1+d_1}|P|$,
then with probability at least $1-\gamma_3$, $V_a$ contains no
element of $P_2$, where $\gamma_3=\epsilon^{d_1}d_4\le \gammabound$
for all $\epsilon\le \epsilon_0$, where $\epsilon_0=({\gammabound}{1\over d_4})^{1/d_1}$.
\item\label{phase2-complexity-statm}
Line \ref{phase2-start-line} to line~\ref{phase2-end-line} iterate at most $\ceiling{\log n\over \log
    (1+\varsigma)}$ times  and generate
at most ${N_2+M\choose M}\cdot {\log n\over \log (1+\varsigma)}$
pairs of centers.
\end{enumerate}
\end{lemma}

\begin{proof}The Lemma is proven with the following cases.

Statement~\ref{early-case1}:    Since $|P_1|\ge |P_2|$, we have $|P_1|\ge {|P|\over 2}$.
Let $p_1={1\over 2}$. With $N_a$ elements from $P$, with
probability at most $\gamma_1=e^{{-\delta_2^2\over
2}p_1N_a}=e^{{-\delta_2^2d_2\over 4}M}\le \gammabound$ (by inequality (\ref{exp-M-ineqn})), there are less than
$(1-\delta_2)p_1N_a={(1-\delta_2)d_2\over 2}M$ elements from
$P_1$ by Theorem~\ref{chernoff3-theorem}.

Statement~\ref{early-case2}:
Let $z_1=\epsilon^{1+d_1}$. By line~\ref{N1-line} of the algorithm, we
have $pN_b={M\over 1-\delta_2}$.  When $N_b$ elements are
selected from $P$, by Theorem~\ref{chernoff3-theorem}, with
probability at most $\gamma_2\le e^{{-\delta_2^2\over
2}z_1N_b}=e^{{-\delta_2^2\over 2}\cdot {M\over 1-\delta_2}}\le
e^{{-\delta_2^2\over 2}\cdot {M\over 0.5}}\le \gammabound$ (by inequality (\ref{exp-M-ineqn}) and the range of $\delta_2$ determined nearby equation (\ref{delta2-d2-ineqn})), multiset $S_b$
has less than $(1-\delta_2)z_1N_b=M$ random points from  $P_2$.

Statement~\ref{phase1-complexity-statm}: After getting $S_a$ and $S_b$ of sizes $N_a$ and $N_b$, respectively,  it takes
${N_a\choose M}\cdot {N_b\choose M}$ cases to enumerate their subsets of size $M$. If $S_a$ contains $M$ elements from $P_1$
and $S_b$ contains $M$ elements from $P_2$, then it generates
${N_a\choose M}\cdot {N_b\choose M}$
pairs of $H_1$ and $H_2$.

Statement~\ref{inbalance-statm}: 
Let $z_2=\epsilon^{1+d_1}$. When $M $ elements are selected in $P$, the probability that $V_a$
contains no element of $P_2$ is at least $(1-z_2)^{M }\ge
1-\epsilon^{1+d_1}M=1-\epsilon^{d_1}d_4$ by
Lemma~\ref{ex-lemma}, and equations (\ref{M-eqn}) and (\ref{c4-eqn}).
Let $\gamma_3=\epsilon^{d_1}d_4$. We have
$\gamma_3\le\gammabound$ for all small positive $\epsilon$ when
$d_1>0$.

Statement~\ref{phase2-complexity-statm}: The loop from line~\ref{phase2-loop-start-line} to
line~\ref{phase2-end-line} iterates at most $t={\log n\over \log
(1+\varsigma)}$ times since $(1+\varsigma)^t\ge n$. Each iteration
of the internal loop from line~\ref{phase2-loop2-start-line} to
line~\ref{phase2-loop2-end-line} generates ${N_2+M\choose M}$ pairs
of centers.
\end{proof}

\begin{lemma}
\label{c1-quality-statm-lemma}
Assume that $V$ only contains elements in $P_i$ $(1\le i\le 2$). Then with
probability at least $1-\gamma_4$ ($\gamma_4\le \gammabound$), the approximate center $c_i$
satisfies the inequality
\begin{eqnarray}
||c(V)-m_i||^2\le {\epsilon(1+\eta )\over \alpha_6}\sigma_i^2
\label{c1-closeness-ineqn}
\end{eqnarray}
\end{lemma}

\begin{proof}
    It follows from
    Lemma~\ref{geometric-random-lemma}.
    Let $\delta^*=\gammabound$. This is because
    \begin{eqnarray}
    \delta^*|M|&=&\gammabound|M|\\
    &=&\gammabound\cdot {d_4\over \epsilon}\\
    &=&{1\over \epsilon}\cdot \left(\gammabound\cdot {12\over \gamma_4}\right)\cdot {(\gamma_4d_4)\over 12}\\
    &\ge& {1\over \epsilon}\cdot  {1\over (1+\eta )}\cdot \alpha_6.
    \end{eqnarray}

    Thus, ${1\over
        \delta^*|M|}\le {\epsilon(1+\eta )\over \alpha_6}$. Therefore,
    the failure probability is at most $\delta^*$ by
    Lemma~\ref{geometric-random-lemma}. Let $\gamma_4=\delta^*$.
\end{proof}

We assume that if the unbalanced condition of Statement
(\ref{inbalance-statm}) of Lemma~\ref{early-sampling-lemma} is satisfied, then inequality (\ref{c1-closeness-ineqn}) holds for $i=1$ with  $V=V_a$ and $c_1=c(V_a)$. In otherwords, inequality $||c_1-m_1||^2\le {\epsilon(1+\eta )\over \alpha_6}\sigma_1^2$ holds at the unbalanced condition since it has a
large probability to be true by Lemma~\ref{c1-quality-statm-lemma} and Statement~\ref{inbalance-statm} of Lemma~\ref{early-sampling-lemma}.

\begin{lemma}\label{f2-upper-lemma}
$f_2(c_1, P_1)\le \left(1+{\epsilon(1+\eta )\over \alpha_6}\right)f_2(m_1,
P_1)$.
\end{lemma}

\begin{proof}
 By Lemma~\ref{KumarSS10-eqn} and
inequality~(\ref{c1-closeness-ineqn}), we have
\begin{eqnarray}
f_2(c_1, P_1)&=&f_2(m_1,P_1)+|P_1|||c_1-m_1||^2\\
&\le&f_2(m_1,P_1)+|P_1|\cdot {\epsilon(1+\eta )\over
\alpha_6}\sigma_1^2\label{f2-c1-P1}\\
&=&f_2(m_1, P_1)+{\epsilon(1+\eta )\over \alpha_6}f_2(m_1, P_1)\label{f2-c1-P2}\\
&=&\left(1+{\epsilon(1+\eta )\over \alpha_6}\right)f_2(m_1, P_1)
\end{eqnarray}
Note that the transition from (\ref{f2-c1-P1}) to (\ref{f2-c1-P2}) is by item~\ref{sigma-j-def} of Definition~\ref{OPT-def} .
\end{proof}




 We discuss the two different cases. They are based on the size of
 $P_2^{out}$.

{\bf Case 1:} $|P_2^{out}|< {\epsilon\over \alpha_1}\beta_2n$.

In this case, we let $c_2=c_1$.

\begin{lemma}\label{approx-m2-lemma}
$||m_2-m_2^{in}||\le \sqrt{\epsilon\over
\alpha_1-\epsilon}\sigma_2$.
\end{lemma}

\begin{proof}
Since  $|P_2|=\beta_2|P|=\beta_2n$, we have $|P_2^{out}|<({\epsilon\over \alpha_1})|P_2|$ by the condition of Case 1. Let $\alpha=1-{\epsilon\over \alpha_1}$. We have $|P_2^{in}|=|P_2-P_2^{out}|=|P_2|-|P_2^{out}|>(1-{\epsilon\over \alpha_1})|P_2|\ge \alpha|P_2|$. By
Lemma~\ref{large-portion-lemma}, we have $||m_2-m_2^{in}||\le
\sqrt{1-\alpha\over \alpha}\sigma_2=\sqrt{{\epsilon\over
\alpha_1}\over 1-{\epsilon\over
\alpha_1}}\sigma_2=\sqrt{\epsilon\over \alpha_1-\epsilon}\sigma_2.$
\end{proof}

\begin{lemma}
    $||m_2^{in}-c_2||\le r_2$.
\end{lemma}

\begin{proof}
    By the definition of $B_2$, we have the following inequalities:
\begin{eqnarray}
||m_2^{in}-c_2||&=&\left|\left|\left({1\over |P_2^{in}|}\sum_{p\in P_2^{in}}p\right)-c_2\right|\right|\\
&=&\left|\left|\left({1\over |P_2^{in}|}\sum_{p\in P_2^{in}}p\right)-\left({1\over |P_2^{in}|}\sum_{p\in P_2^{in}}c_2\right)\right|\right|\\
&\le&{1\over |P_2^{in}|}\left|\left|\sum_{p\in P_2^{in}}(p-c_2)\right|\right|\\
&\le&{1\over |P_2^{in}|}\sum_{p\in P_2^{in}}||p-c_1||\\
&\le& r_2.
\end{eqnarray}
\end{proof}

\begin{lemma}\label{case1-f2-lemma}
$f_2(c_2, P_2)\le f_2(m_2, P_2)(1+{2\epsilon\over
\alpha_1-\epsilon})+{2\epsilon\over \alpha_5}OPT_2(P)$.
\end{lemma}

\begin{proof}
\begin{eqnarray}
f_2(c_2, P_2)&=&f_2(m_2, P_2)+|P_2|||m_2-c_2||^2\label{f2-upper-a1-ineqn}\\
&\le&f_2(m_2, P_2)+|P_2|(2||m_2-m_2^{in}||^2+2||m_2^{in}-c_2||^2)\label{f2-upper-a1b-ineqn}\\
&\le&f_2(m_2, P_2)+|P_2|\left(2\left(\sqrt{\epsilon\over \alpha_1-\epsilon}\sigma_2\right)^2+2r_2^2\right)\label{f2-upper-a2-ineqn}\\
&=&f_2(m_2, P_2)+|P_2|\left({2\epsilon\over \alpha_1-\epsilon}\right)\cdot \sigma_2^2+2|P_2|r_2^2\\
&=&\left(1+{2\epsilon\over \alpha_1-\epsilon}\right)f_2(m_2, P_2)+{2\epsilon\over \alpha_5}OPT_2(P).
\end{eqnarray}

The transition from inequality (\ref{f2-upper-a1-ineqn}) to
inequality (\ref{f2-upper-a2-ineqn}) is by
Lemma~\ref{approx-m2-lemma}.
\end{proof}

\begin{lemma}
$f_2(c_1, P_1)+f_2(c_2, P_2)\le (1+\epsilon)OPT_2(P)$.
\end{lemma}

\begin{proof}  By  Lemma~\ref{f2-upper-lemma} and Lemma~\ref{case1-f2-lemma},  we have inequalities:
\begin{eqnarray}
&&f_2(c_1, P_1)+f_2(c_2, P_2)\\
&\le&\left(1+{\epsilon(1+\eta )\over
\alpha_6}\right)f_2(m_1,
P_1)+\left(1+{2\epsilon\over \alpha_1-\epsilon}\right)f_2(m_2, P_2)+{2\epsilon\over \alpha_5}OPT_2(P)\label{f2-c1-P1-ineqn1}\\
&\le&\left(1+{2\epsilon\over \alpha_5}+\max\left({\epsilon(1+\eta )\over \alpha_6},{2\epsilon\over
\alpha_1-\epsilon}\right)\right)OPT_2(P)\label{f2-c1-P1-ineqn2}\\
&\le&\left(1+{\epsilon\over 2}+\max\left({\epsilon\over 4},{\epsilon\over
    2}\right)\right)OPT_2(P)\label{f2-c1-P1-ineqn2b}\\
 &\le& (1+\epsilon)OPT_2(P).\label{f2-c1-P1-ineqn3}
\end{eqnarray}
The transition (\ref{f2-c1-P1-ineqn2}) to (\ref{f2-c1-P1-ineqn3}) is
based on equation (\ref{alpha6-eqn}) and the setting of parameters  in
Section~\ref{parameters-sect}. Note that ${(1+\eta )\over \alpha_6}\le {1.5\over 3}={1\over 2}$ by inequality  (\ref{constant-start-conditions}) and equation (\ref{eta-eqn}), and ${2\over \alpha_1-\epsilon}\le {1\over 2}$ by equation (\ref{alpha1-eqn}).
\end{proof}

{\bf Case 2:} $|P_2^{out}|\ge {\epsilon\over \alpha_1}\beta_2n$.

\begin{lemma}\label{lower-bound-lemma}
${|P_2^{out}|\over |P-B_2|}\ge {\epsilon^2\over \alpha_2}$ for all
 positive $\epsilon\le \epsilon_1$, where $\epsilon_1=\delta_6$.
\end{lemma}

\begin{proof}    Note that $\epsilon< 1$ by our assumption in Section~\ref{parameters-sect}.
Assume that ${|P_2^{out}|\over |P-B_2|}< {\epsilon^2\over
\alpha_2}$.  This implies $|P-B_2|>
{\alpha_2|P_2^{out}|\over \epsilon^2}$. We will derive a contradiction.  We have inequalities:
\begin{eqnarray}
|P-B_2|- |P_2^{out}|&\ge&  {\alpha_2|P_2^{out}|\over \epsilon^2}-|P_2^{out}|\\
&=&{(\alpha_2-\epsilon^2)|P_2^{out}|\over \epsilon^2}\\
&\ge&{(\alpha_2-\epsilon^2)\epsilon\beta_2n \over \alpha_1\epsilon^2}\\
&=&{(\alpha_2-\epsilon^2)\beta_2n \over \alpha_1\epsilon}.
\end{eqnarray}

On the other hand, it is easy to see that $P_2\subseteq (B_2\cup
P_2^{out})$ and $((P-B_2)- P_2^{out})\subseteq P_1$. For each element
$x\in ((P-B_2)- P_2^{out})$, we have $||c_1-x||\ge r_2$. We also
have inequalities:
\begin{eqnarray}
f_2(c_1, P_1)&>&    (|P-B_2|- |P_2^{out}|)r_2^2\\
&\ge&  {(\alpha_2-\epsilon^2)\beta_2n \over \alpha_1\epsilon}\cdot
\left(\sqrt{\epsilon\over
\alpha_5\beta_2}\sigma_{opt}\right)^2\\
&\ge&  {(\alpha_2-\epsilon^2)\beta_2n \over \alpha_1\epsilon}\cdot
{\epsilon\over
\alpha_5\beta_2}\sigma_{opt}^2\\
&\ge&  {(\alpha_2-\epsilon^2)n \over \alpha_1\alpha_5}\sigma_{opt}^2\\
&\ge&  {(\alpha_2-\epsilon^2) \over \alpha_1\alpha_5}OPT_2(P)\\
&>& 2OPT_2(P)\ \ \ ({\rm by~ equation~\ref{alpha2-eqn} \ and }\ \epsilon\le \delta_6).
\end{eqnarray}

This contradicts Lemma~\ref{f2-upper-lemma} since $\left(1+{\epsilon(1+\eta )\over \alpha_6}\right)<2$ and $f_2(m_1,P_1)\le OPT_2(P)$.
\end{proof}

\begin{lemma}\label{m2-m2-bar-lemma}
$||m_2-\tilde{m_2}||\le (1-{\epsilon\over \alpha_1})r_2$.
\end{lemma}

\begin{proof}
\begin{eqnarray}
||m_2-\tilde{m_2}||&=&\left|\left|{1\over |P_2|}\sum_{p\in
P_2}p-{1\over
|P_2|}\left(\sum_{p\in P^{in}_2}\tilde{p}+\sum_{p\in P^{out}_2}p\right)\right|\right|\\
&=&{1\over |P_2|}||\sum_{p\in P^{in}_2}(p-\tilde{p})||\\
&\le&{1\over |P_2|}\sum_{p\in P^{in}_2}||p-\tilde{p}||\\
&\le&{|P^{in}_2|\over |P_2|}\cdot r_2\\
&=&{|P_2|-|P_2^{out}|\over |P_2|}\cdot r_2\\
&\le&{|P_2|-{\epsilon\over \alpha_1}\beta_2n \over |P_2|}\cdot r_2\\
&=&{|P_2|-{\epsilon\over \alpha_1}|P_2| \over |P_2|}\cdot r_2\\
 &\le& \left(1-{\epsilon\over \alpha_1}\right)r_2.
\end{eqnarray}

\end{proof}

\begin{lemma}\label{f2-alpha6-lemma}
$f_2(\tilde{m_2},\tilde{P_2})\le 2f_2(m_2, P_2)+\alpha_6 \beta_2nr_2^2$.
\end{lemma}

\begin{proof}
We have the following inequalities:
\begin{eqnarray}
f_2(\tilde{m_2},\tilde{P_2})&\le&\sum_{p\in P^{in}_2} ||\tilde{p}-\tilde{m_2}||^2+\sum_{p\in P^{out}_2} ||p-\tilde{m_2}||^2\label{f2-m2-P2-inequality1}\\
&\le&\sum_{p\in P^{in}_2} ||\tilde{p}-p+p-\tilde{m_2}||^2+\sum_{p\in P^{out}_2} ||p-\tilde{m_2}||^2\label{f2-m2-P2-inequality2}\\
&\le&\sum_{p\in P^{in}_2} 2(||\tilde{p}-p||^2+||p-\tilde{m_2}||^2)+\sum_{p\in P^{out}_2} ||p-\tilde{m_2}||^2\label{f2-m2-P2-inequality3}\\
&\le&2\sum_{p\in P^{in}_2}||\tilde{p}-p||^2+2\sum_{p\in P_2} ||p-\tilde{m_2}||^2\label{f2-m2-P2-inequality4}\\
&\le&2|P^{in}_2|r_2^2+2f_2(\tilde{m_2},P_2)\label{f2-m2-P2-inequality5}\\
&\le&2|P^{in}_2|r_2^2+2f_2(m_2,P_2)+2|P_2|\cdot||\tilde{m_2}-m_2||^2\label{f2-m2-P2-inequality6}\\
     &\le& 2f_2(m_2, P_2)+4 \beta_2nr_2^2\label{f2-m2-P2-inequality7}\\
     &\le& 2f_2(m_2, P_2)+\alpha_6 \beta_2nr_2^2. \label{f2-m2-P2-inequality8}
\end{eqnarray}
The transition from (\ref{f2-m2-P2-inequality5}) to
(\ref{f2-m2-P2-inequality6}) is by Lemma~\ref{KumarSS10-eqn}. The
transition from (\ref{f2-m2-P2-inequality6}) to
(\ref{f2-m2-P2-inequality7}) is by Lemma~\ref{m2-m2-bar-lemma}.
 The
transition from (\ref{f2-m2-P2-inequality7}) to
(\ref{f2-m2-P2-inequality8}) is based on inequality
(\ref{constant-start-conditions}).
\end{proof}

\begin{lemma}\label{later-sampling-lemma}  The iteration from line~\ref{phase2-loop-start-line} to line~\ref{phase2-end-line} of the algorithm has the property:
\begin{enumerate}[1.]
    \item\label{iteration-first}
    There is an integer $j$ $\left(0\le j\le \ceiling{\log n\over \log
        (1+\varsigma)}\right)$ such that $(P-B_2)\subseteq Q_j$ and $|Q_j|\le
    (1+\varsigma)|P-B_2|$.

\item\label{iteration-second}
                    With probability at least
$1-\gamma_5$, when $(P-B_2)\subseteq Q_j$ and $|Q_j|\le
(1+\varsigma)|P-B_2|$,
           the algorithm $2$-Means(.) at line~\ref{H'-center-line} generates $c_2=c(H')$ for a subset $H'$ of size $M$ such that
\begin{eqnarray}
 ||\tilde{m_2}-c_2||^2\le {\epsilon\over \alpha_6}{f_2(\tilde{m_2},\tilde{P_2})\over
 |\tilde{P_2}|},\label{c2-closeness-ineqn}
 \end{eqnarray}
 where $\gamma_5\le \gammabound$ for all positive $\epsilon\le \epsilon_2$ with $\epsilon_2={\delta_1^2d_4\over 2(1-\delta_1)\ln {12\over \gamma}}$.
\end{enumerate}
\end{lemma}

\begin{proof}

Statement~\ref{iteration-first}: Since $Q_0=P$,  it is trivial $(P-B_2)\subseteq Q_0$. By Statement~\ref{phase2-complexity-statm} of Lemma~\ref{early-sampling-lemma}, the variable $j$  is in the range  $[0,\ceiling{\log n\over \log
    (1+\varsigma)}]$.
For each $j$ during the iteration, we have $Q_{j+1}\subseteq Q_{j}$, and the size of $Q_{j+1}$ is reduced by a factor $(1+\varsigma)$. Furthermore, $Q_{j+1}$ keeps the elements of $Q_j$ with $\dist(q,c_1)\ge m_j$. Therefore, there is an integer $j$ such that $(P-B_2)\subseteq Q_j$ and $|Q_j|\le
(1+\varsigma)|P-B_2|$.

Statement~\ref{iteration-second}: Assume that $(P-B_2)\subseteq Q_j$ and $|Q_j|\le
(1+\varsigma)|P-B_2|$.
By Lemma~\ref{lower-bound-lemma}, with probability at least
$z_3={\epsilon^2\over \alpha_2(1+\varsigma)}$,
 a random element in $Q_j$ is in $P_2^{out}$. By line~\ref{N3-line} of the algorithm $2$-Means(.), we have $z_3N_2={M\over 1-\delta_1}$.

 We assume  $\epsilon\in \left(0, {\delta_1^2d_4\over 2(1-\delta_1)\ln {12\over \gamma}}\right]$.
  Let $a={\delta_1^2\over 2(1-\delta_1)}$.
 If $\epsilon\le  {ad_4\over \ln {12\over \gamma}}={\delta_1^2d_4\over 2(1-\delta_1)\ln {12\over \gamma}}$, then
 $e^{-aM}=e^{-\delta_1^2M\over 2(1-\delta_1)}\le {\gamma\over 12}$  by equation~(\ref{exp-M2-ineqn}) of Lemma~\ref{parameters-lemma}.
  When $N_2$ random
 elements are chosen from $Q_j$, by Theorem~\ref{chernoff3-theorem},
 with probability at most $e^{-\delta_1^2pN_2\over 2}= e^{-\delta_1^2M\over 2(1-\delta_1)}=e^{-aM}\le {\gamma\over 12}$ (for all positive $\epsilon\le  {\delta_1^2d_4\over 2(1-\delta_1)\ln {12\over \gamma}}$),
 there are less than $(1-\delta_1)pN_2=M$ elements from
 $P_2^{out}$.
 Let  $\gamma_{5,a}=e^{-\delta_1^2z_3N_2\over 2}$. We have $\gamma_{5,a}\le {\gamma\over 12}$.

   When $M$ random points, which form subset $H'$ at line~\ref{phase2-loop2-start-line} of $2$-Means(.), are chosen from $\tilde{P_2}=P_2^{out}\cup
 \tilde{P_2^{in}}$, we can get $c_2=c(H')$ that satisfies inequality (\ref{c2-closeness-ineqn}) by Lemma~\ref{geometric-random-lemma}.
We have $\gamma_{5,b}M={\gamma_{5,b}d_4\over 12\epsilon}={\alpha_6\over \epsilon}$  (see equqtion (\ref{gamma5-b-eqn}) for $\gamma_{5,b}$).  We have ${1\over \gamma_{5,b}M}= {\epsilon\over \alpha_6}$.  Thus, we have inequality~(\ref{c2-closeness-ineqn}) with failure probability at most $\gamma_{5,b}$.
It is easy to see that $\gamma_5=\gamma_{5,a}+\gamma_{5,b}\le \gammabound$.
\end{proof}

Assume that $c_2$ satisfies the inequality
(\ref{c2-closeness-ineqn}). We note $|P_2|=|\tilde{P_2}|$.

    \begin{lemma}\label{f2-f2-opt2-lemma}
    $f_2(c_2,P_2)\le (1+{4\epsilon\over \alpha_6})f_2(m_2, P_2)+({2\epsilon\over \alpha_5}+ {2\epsilon^2\over \alpha_5})OPT_2(P)$.
    \end{lemma}

    \begin{proof}
We have the inequalities:
\begin{eqnarray}
f_2(c_2,P_2)&=& f_2(m_2, P_2)+|P_2|||m_2-c_2||^2\label{f2-c2-P2-ineqn1}\\
&\le&  f_2(m_2, P_2)+|P_2|(2||m_2-\tilde{m_2}||^2+2||\tilde{m_2}-c_2||^2)\label{f2-c2-P2-ineqn2}\\
&\le&  f_2(m_2, P_2)+2|P_2|r_2^2+2|P_2|||\tilde{m_2}-c_2||^2)\label{f2-c2-P2-ineqn3}\\
&\le&  f_2(m_2, P_2)+2|P_2|r_2^2+2|P_2|\cdot {\epsilon\over \alpha_6}\cdot {f_2(\tilde{m_2},\tilde{P_2})\over |P_2|}\label{f2-c2-P2-ineqn4}\\
&=&  f_2(m_2, P_2)+2|P_2|r_2^2+{2\epsilon\over \alpha_6}f_2(\tilde{m_2},\tilde{P_2})\label{f2-c2-P2-ineqn5}\\
&=&  f_2(m_2, P_2)+2|P_2|\left(\sqrt{\epsilon\over \alpha_5\beta_2}\sigma_{opt}\right)^2+{2\epsilon\over \alpha_6}f_2(\tilde{m_2},\tilde{P_2})\label{f2-c2-P2-ineqn6}\\
&=&  f_2(m_2, P_2)+\left({2|P_2|\epsilon\over \alpha_5\beta_2}\right)\sigma_{opt}^2+{2\epsilon\over \alpha_6}f_2(\tilde{m_2},\tilde{P_2})\label{f2-c2-P2-ineqn7}\\
&\le&  f_2(m_2, P_2)+\left({2|P_2|\epsilon\over \alpha_5\beta_2}\right)\sigma_{opt}^2+{2\epsilon\over \alpha_6}\left(2f_2(m_2, P_2)+\alpha_6 \beta_2nr_2^2\right)\label{f2-c2-P2-ineqn8}\\
&=&  f_2(m_2, P_2)+\left({2|P_2|\epsilon\over \alpha_5\beta_2}\right)\sigma_{opt}^2+{2\epsilon\over \alpha_6}\left(2f_2(m_2, P_2)+\alpha_6 \beta_2n\left(\sqrt{\epsilon\over \alpha_5\beta_2}\sigma_{opt}\right)^2\right)\label{f2-c2-P2-ineqn9}\\
&=&  f_2(m_2, P_2)+\left({2|P_2|\epsilon\over \alpha_5\beta_2}\right)\sigma_{opt}^2+{2\epsilon\over \alpha_6}\left(2f_2(m_2, P_2)+\alpha_6 n\left({\epsilon\over \alpha_5}\right)\sigma_{opt}^2\right)\label{f2-c2-P2-ineqn10}\\
&=&  f_2(m_2, P_2)+\left({2\epsilon\over \alpha_5}\right)OPT_2(P)+{2\epsilon\over \alpha_6}\left(2f_2(m_2, P_2)+\alpha_6 \left({\epsilon\over \alpha_5}\right)OPT_2(P)\right)\label{f2-c2-P2-ineqn11}\\
&=&  f_2(m_2, P_2)+{4\epsilon\over \alpha_6}f_2(m_2, P_2)+{2\epsilon\over \alpha_5}OPT_2(P)+ {2\epsilon^2\over \alpha_5}OPT_2(P)\label{f2-c2-P2-ineqn12}\\
&=&  \left(1+{4\epsilon\over \alpha_6}\right)f_2(m_2,
P_2)+\left({2\epsilon\over \alpha_5}+ {2\epsilon^2\over
\alpha_5}\right)OPT_2(P).\label{f2-c2-P2-ineqn13}
\end{eqnarray}

The transition from (\ref{f2-c2-P2-ineqn1}) to
(\ref{f2-c2-P2-ineqn2}) is by Lemma~\ref{norm-square-lemmma}.
The
transition from (\ref{f2-c2-P2-ineqn2}) to (\ref{f2-c2-P2-ineqn3})
is by Lemma~\ref{m2-m2-bar-lemma}.
 The
transition from (\ref{f2-c2-P2-ineqn3}) to (\ref{f2-c2-P2-ineqn4})
is by inequality (\ref{c2-closeness-ineqn}). 
The transition from
(\ref{f2-c2-P2-ineqn5}) to (\ref{f2-c2-P2-ineqn6}) is by
item~\ref{rj-def} of Definition~\ref{rj-Pj-mj-def}. The transition
from (\ref{f2-c2-P2-ineqn7}) to (\ref{f2-c2-P2-ineqn8}) is by
Lemma~\ref{f2-alpha6-lemma}. The transition from
(\ref{f2-c2-P2-ineqn8}) to (\ref{f2-c2-P2-ineqn9}) is by
item~\ref{rj-def} of Definition~\ref{rj-Pj-mj-def}. The transition
from (\ref{f2-c2-P2-ineqn9}) to (\ref{f2-c2-P2-ineqn10}) is by
item~\ref{sigma-opt2-def} of Definition~\ref{OPT-def}. We note that
$|P|={|P_2|\over \beta_2}$.
    \end{proof}

  \begin{lemma}\label{app-ratio-lemma}
  $f_2(c_1, P_2)+f_2(c_2,P_2)\le (1+\epsilon)OPT_2(P)$ for all positive $\epsilon\le \delta_6/4$.
  \end{lemma}
\begin{proof}
We have the following inequalities:
\begin{eqnarray}
&&f_2(c_1, P_2)+f_2(c_2,P_2)\label{f2-f2-OPT2-ineqn0}\\
&\le& \left(1+{\epsilon(1+\eta )\over \alpha_6}\right)f_2(m_1,
P_1)+\left(1+{4\epsilon\over \alpha_6}\right)f_2(m_2,
P_2)+\left({2\epsilon\over \alpha_5}+ {2\epsilon^2\over
\alpha_5}\right)OPT_2(P)\label{f2-f2-OPT2-ineqn1}\\
&\le& \left(1+{4\epsilon(1+\eta )\over \alpha_6}\right)(f_2(m_1,
P_1)+f_2(m_2, P_2))+\left({2\epsilon\over \alpha_5}+
{2\epsilon^2\over
\alpha_5}\right)OPT_2(P)\label{f2-f2-OPT2-ineqn2}\\
&\le& \left(1+{4\epsilon(1+\eta )\over
\alpha_6}\right)OPT_{2}(P)+\left({2\epsilon\over \alpha_5}+
{2\epsilon^2\over
\alpha_5}\right)OPT_2(P)\\
&\le& \left(1+{4\epsilon(1+\eta )\over \alpha_6}+{2\epsilon\over
\alpha_5}+ {2\epsilon^2\over
\alpha_5}\right)OPT_2(P)\label{ratio-last1a-ineqn}\\
&\le& \left(1+\left({4(1+\eta )\over \alpha_6}+{2+2\epsilon\over
\alpha_5}\right)\epsilon \right)OPT_2(P)\label{ratio-last1b-ineqn}\\
&\le& \left(1+\left({4\over \alpha_6}+{2\over
\alpha_5}\right)(1+\eta+2\epsilon )\epsilon\right)OPT_2(P)\label{ratio-last1-ineqn}\\
&\le& \left(1+\left(1-\delta_6\right)(1+\delta_6/2+2\epsilon )\epsilon\right)OPT_2(P)\label{ratio-last1b-ineqn}\\
&\le& \left(1+\left(1-\delta_6\right)(1+\delta_6)\epsilon\right)OPT_2(P)\label{ratio-last1c-ineqn}\\
&\le& (1+\epsilon)OPT_2(P)\label{ratio-last2-ineqn}\ \ \ {\rm for\
all\ small\ positive\ } \epsilon\le \delta_6/4.
\end{eqnarray}
The transition from (\ref{f2-f2-OPT2-ineqn0}) to
(\ref{f2-f2-OPT2-ineqn1}) is by Lemma~\ref{f2-upper-lemma} and
Lemma~\ref{f2-f2-opt2-lemma}. The transition from
(\ref{ratio-last1-ineqn}) to (\ref{ratio-last2-ineqn}) is by
equation (\ref{constant-end-conditions}).
\end{proof}

\begin{lemma}\label{failure-approx-lemma} There is a positive constant $\delta_0$ such that
with probability at most $1-\gamma$, the algorithm $2$-Means(.)
returns a set $U$ that contains at least one  $2$ centers to induce a $(1+\epsilon)$-approximation for the constrained $2$-means  problem for every $\epsilon\in (0,\delta_0]$.
\end{lemma}

\begin{proof} Let $\delta_0=\min(\epsilon_0, \epsilon_1,\epsilon_2)$, where $\epsilon_0$, $\epsilon_1$ and $\epsilon_2$ are defined in Lemmas~\ref{early-sampling-lemma},~\ref{lower-bound-lemma} and~\ref{later-sampling-lemma}, respectively. By Lemma~\ref{vibrate-lemma}, Lemma~\ref{later-sampling-lemma}, and
Lemma~\ref{early-sampling-lemma},  the probability of the algorithm
to fail is at most $\gamma_1+\gamma_2+\gamma_3+\gamma_4+\gamma_5\le
\gamma$ from union bound, where the parameters
$\gamma_1,\gamma_2,\gamma_3,\gamma_4$ and $\gamma_5$ are defined in
Lemmas~\ref{early-sampling-lemma} and~\ref{later-sampling-lemma}.
The $(1+\epsilon)$ approximation follows from
Lemma~\ref{app-ratio-lemma}.
\end{proof}

\begin{lemma}\label{time-lemma}
The algorithm 2-Means(.) runs in
$\bigO\left(dn+d\left({e^2(d_2 +\delta)\over
\epsilon^{1+d_1}}\right)^{M}+\left(({\alpha_2+\delta\over
\epsilon^2})e\right)^{M}\cdot {\log n\over \log
(1+\varsigma)})\right)$ -time.
\end{lemma}

\begin{proof}By
Statement~\ref{phase1-complexity-statm} of
Lemma~\ref{early-sampling-lemma}, the number of pairs of centers
generated by line \ref{phase1-start-line} to
line~\ref{phase1-end-line} of the algorithm $2$-Means(.) is at most
\begin{eqnarray}
T_1&=&{N_a\choose M}\cdot {N_b\choose
M}\label{T1-inequality1}\\
&\le&\left(\left({N_a^M\over M!}\right)\cdot
\left({N_b^M\over
M!}\right)\cdot \right)\label{T1-inequality2}\\
&=&\bigO\left(\left({N_a\cdot e\over M}\right)^M\cdot
\left({N_b\cdot e\over
M}\right)^M\right)\label{T1-inequality3}\\
&=&\bigO\left((ed_2)^M\cdot \left({e\over (1-\delta_2)\cdot
\epsilon^{1+d_1}}\right)^M\right)\label{T1-inequality4}\\
&=&\bigO\left( \left({d_2e^2 \over (1-\delta_2)\cdot
\epsilon^{1+d_1}}\right)^M\right)\label{T1-inequality5}\\
&=&\bigO\left( \left({(1+2\delta_2)\cdot d_2e^2 \over
\epsilon^{1+d_1}}\right)^{M}\right)\label{T1-inequality7}\\
&=&\bigO\left( \left({e^2(d_2 +\delta)\over
\epsilon^{1+d_1}}\right)^{M}\right).\label{T1-inequality}
\end{eqnarray}

The transition from (\ref{T1-inequality7}) to (\ref{T1-inequality})
is by inequality~(\ref{delta2-d2-ineqn}).

By Statement~\ref{phase2-complexity-statm} of
Lemma~\ref{early-sampling-lemma}, the number of pairs of centers
generated by line \ref{phase2-start-line} to
line~\ref{phase2-end-line} is at most
\begin{eqnarray}
T_2&=&{N_2+M\choose M}\cdot \ceiling{\log n\over \log (1+\varsigma)}\label{T2-ineqn1}\\
&\le&\left({(N_2+M)^M\over M!}\cdot \ceiling{\log
n\over \log (1+\varsigma)}\right)\label{T2-ineqn2}\\
&=&\bigO\left(({(N_2+M)e\over M})^M\cdot {\log
n\over \log (1+\varsigma)}\right)\label{T2-ineqn3}\\
&=&\bigO\left(\left({((1+\varsigma)\cdot {\alpha_2\over
\epsilon^2}\cdot {1\over 1-\delta_1}+1)e}\right)^M\cdot  {\log
n\over \log (1+\varsigma)}\right)\label{T2-ineqn4}\\
&=&\bigO\left(\left({({(1+\varsigma)(1+2\delta_1)\alpha_2\over
\epsilon^2}+1)e}\right)^M\cdot  {\log
n\over \log (1+\varsigma)}\right)\label{T2-ineqn5}\\
&=&\bigO\left(\left({\left({\alpha_2+\delta/2\over
\epsilon^2}+1\right)e}\right)^M\cdot {\log
n\over \log (1+\varsigma)}\right)\label{T2-ineqn6}\\
&=&\bigO\left(\left({\left({\alpha_2+\delta/2+\epsilon^2\over
\epsilon^2}\right)e}\right)^M\cdot {\log
n\over \log (1+\varsigma)}\right)\label{T2-ineqn7}\\
&=&\bigO\left(\left({\alpha_2+\delta\over
\epsilon^2})e\right)^{M}\cdot {\log n\over \log
(1+\varsigma)}\right)\ \ \ {\rm for\ all\ small\ positive\ }
\epsilon\ {\rm with}\ \epsilon^2\le \delta/2.\label{T2-inequality}
\end{eqnarray}

The transition (\ref{T2-ineqn5}) to (\ref{T2-ineqn6}) is by
inequality (\ref{varsigma-delta1-ineqn}).
It takes $\bigO(dn)$ time to find the $i$-th largest element after
$c_1$ is fixed. The number of elements goes down by a constant factor. Therefore, the total time to find the $i$-th largest element is still  $\bigO(dn)$.
Therefore, we have total time $\bigO(dn+d(T_1+T_2))$.
\end{proof}

\begin{theorem}\label{integer-theorem}For the constrained $2$-means  problem, there is a positive constant $\delta_0$ such that
there is a randomized $\bigO\left(dn+d
\left(({\alpha_2+\delta\over
\epsilon^2})e\right)^{d_4\over\epsilon}\cdot \log n)\right)$-time
 algorithm for every $\epsilon\in (0,\delta_0]$. Moreover, it outputs a collection $U$ of approximate center pairs $(c_1, c_2)$ such that one of pairs in $U$ can induce a $(1+\epsilon)$-approximation for the constrained $2$-means problem.
\end{theorem}

\begin{proof}
The time complexity of the algorithm $2$-Means(.) follows from
Lemma~\ref{time-lemma} by selecting $\epsilon$ and the parameters
mentioned in Section~\ref{parameters-sect} to be small enough. The
failure probability of the algorithm and its approximation ratio
follow from Lemma~\ref{failure-approx-lemma}.
\end{proof}

\begin{corollary}\label{integer1-corollary}For the constrained $2$-means  problem, there is a positive constant $\delta_0$ such that
there is a $\bigO(dn+d({1\over\epsilon})^{\bigO({1\over
\epsilon})}\log n)$-time algorithm for every $\epsilon\in (0,\delta_0]$. Moreover, it outputs a collection $U$ of approximate center pairs $(c_1, c_2)$ such that one of pairs in $U$ can induce a $(1+\epsilon)$-approximation for the constrained $2$-means.
\end{corollary}

\section{Improving the Existing PTAS for Constrained
$k$-Means}\label{k-mean-sect}

In this section, we generalize the method developed in this paper,
and derive improved PTAS for the constrained $k$-means problems. We
observed that all the existing PTAS can be transformed in our
framework.

\newcommand {\extendKM}{{\rm $k$-means extension}}
\newcommand {\gammaboundk}{{\mu_2\over 4}}

\begin{definition}\label{extend-k-means-def}
Let $P$ be an input of points in $\mathbb{R}^d$ for a constrained $k$-means problem, and $P_1,P_2,\cdots, P_k$ be the $k$ clusters
of $P$ with $|P_1|\ge |P_2|\ge \cdots\ge |P_k|$.    Let $\delta(.):(0,1)\rightarrow (0,1)$ be an nonincreasing
function with $\lim_{\epsilon\rightarrow 0}\delta(\epsilon)=0$.
 A
$(\epsilon,\delta(.),\mu)$-\extendKM~is an algorithm $A(k,
P,\epsilon, T)$ that gets a subset $T=(c_1,\cdots, c_r)$ approximate
centers for $P_1,\cdots, P_r$, respectively, with
$||c_j-c(P_j)||\le \delta(\epsilon)\sigma_j$ $(1\le j\le r)$
for the largest $r$
clusters $P_1,\cdots, P_r$,  and returns the rest $k-r$ approximate centers
$c_{r+1},\cdots, c_k$ for $P_{r+1},\cdots, P_k$, respectively, such
that they form a $(1+\epsilon)$-approximation for the $k$-means
problem. Furthermore, the failure probability is at most $\mu$.

A \extendKM~$A(k, P,\epsilon, T)$ has time complexity $Z(k,r,n,
d,\epsilon)$ that is the time to
find the rest $k-r$ centers, where $T=(c_1,\cdots, c_r)$ provides the algorithm $A(.)$ $r$ approximate
centers for $P_1,\cdots, P_r$, respectively.

A \extendKM~$A(k, P,\epsilon, T)$ has pair complexity $H(k,r,n,
d,\epsilon)$ to be an upper bound of the number of $k$-centers in
its output list after it gets $k-r$ approximate centers.
\end{definition}

Our method is based on the two cases that were discussed in the
algorithm for the constrained $2$-means problem. The existing approximation algorithm for the constrained $k$-means problem can be speed up.

 In the case 1, the two largest clusters $P_1$ and $P_2$ are balanced. We will efficiently find the approximate centers $c_1$ and $c_2$ for $P_1$ and $P_2$, respectively.  The
 centers for $P_3,\cdots, P_k$ can be found by calling $A(.)$.

 In the case 2, the two largest clusters $P_1$ and $P_2$  are not balanced. In this case, $P_1$ is much larger the union of the other clusters.   We will efficiently find the approximate centers $c_1$ for $P_1$. The
 centers for $P_2,\cdots, P_k$ can be found by calling $A(.)$.

\vskip 20pt {\bf Algorithm} $k$-Means$(P, \epsilon, A(.))$

Input: $P$ is a set of points, a real number $\epsilon$ is in $
(0,1)$, and $A(.)$ is a $(\epsilon,\delta(.),\mu_1)$ \extendKM~with
a fixed $\mu_1<1$.

Output: A collection $U$ of $k$ centers $(c_1,c_2,\cdots, c_k)$.

\begin{enumerate}[1.]

\item
\qquad Let $U=\emptyset$;

\item
\qquad Let $D=\emptyset$;

\item
\qquad Choose parameters $\mu_2=\mu_3$ such that
$\mu_1+\mu_2+\mu_3<1$;

\item
\qquad Let $\gamma^*={\mu_3\over 2}$;

\item\label{k-means-start-line}
\qquad Let $\delta_2$,  $d_1$ and $d_2$  be the same as those in the algorithm
$2$-Means(.);


\item
\qquad Let $M={1\over \gamma^*\delta(\epsilon)}$;

\item\label{EN1a-line}
\qquad Let $N_a=d_2kM$;

\item\label{EN1b-line}
\qquad Let $N_b={M(k-1)\over (1-\delta_2)\cdot
\delta(\epsilon)^{1+d_1}}$;

\item\label{Ephase1-start2-line}
\qquad Select a set $S_a$ of $N_a$ random samples from $P$;

\item
\qquad Select a set $S_b$ of $N_b$ random samples from $P$;

\item
\qquad For every two subsets $H_1$ of $S_a$ and $H_2$ of
$S_b$  of size $M$,

\item
\qquad \{

\item\label{two-centers-generation2-line}
\qquad\qquad Compute the centroid $c(H_1)$ of $H_1$, and $c(H_2)$ of
$H_2$;


\item\label{add-two-center-line}
\qquad\qquad Let $D=D\cup  \{(c(H_1), c(H_2))\}$.

\item\label{Ephase1-end2-line}
\qquad \}

\item\label{Ephase2-start-line}
\qquad Select a set $V_a$ of $M$ random samples from $P$;

\item
\qquad Compute the centroid $c_1=c(V_a)$ of $S$;


\item\label{k-means-end-line}
 \qquad $D=D\cup  \{(c_1)\}$;


\item
\qquad For each $T\in D$, let $U=U\cup A(k, P,\epsilon, T)$;

\item
\qquad Output $U$;

\end{enumerate}

{\bf End of Algorithm}

\vskip 20pt

The following Lemma~\ref{Eearly-sampling-lemma2} is similar to Lemma~\ref{early-sampling-lemma}. There are some modifications with the parameters.

\begin{lemma}\label{Eearly-sampling-lemma2} The algorithm $k$-Means(.) has the following
properties:
\begin{enumerate}[1.]
\item\label{second-early1}
With probability at least $1-\gamma_1$, at least $(1-\delta_2)d_2M$
random points are from $P_1$ in $S_a$, where
$\gamma_1=2^{{-\delta_2^2d_2\over 4}M}\le \gammaboundk$ for all $\epsilon\le \epsilon_0$, where $\epsilon_0$ is fixed in $(0,1)$. \item\label{second-early2}
If the two clusters satisfy ${\delta(\epsilon)^{1+d_1}\over
k-1}|P|\le |P_2|\le |P_1|$, then with probability at least
$1-\gamma_2$,  at least $M$ random points are from $P_2$ in
$S_b$, where $\gamma_2=2^{{-\delta_2^2\over 2}\cdot {M\over
1-\delta_2}}\le \gammaboundk$  for all $\epsilon\le \epsilon_0$, where $\epsilon_0$ is fixed in $(0,1)$.
\item\label{second-early3}
Line~\ref{Ephase1-start2-line} to line~\ref{Ephase1-end2-line} of the algorithm $k$-Means(.)  generate at most ${N_a\choose M}\cdot
{N_b\choose M}=\bigO\left( \left({(1+2\delta_2)\cdot d_2e^2k(k-1) \over
    \delta(\epsilon)^{1+d_1}}\right)^{M}\right)$ pairs of centers, where $e\approx 2.719$.
\item\label{second-early4}
If the clusters $P_1$ and $P_2$ satisfy
$|P_2|<{\delta(\epsilon)^{1+d_1}\over k-1}\cdot |P|$, then with
probability at least $1-\gamma_3$, $V_a$ contains no element of
$P_2$, where $\gamma_3=\delta(\epsilon)^{1+d_1}M\le \gammaboundk$
for all $\epsilon\le \epsilon_1$, where $\epsilon_1$ is fixed in $(0,1)$.

\item\label{second-early5}
Assume that $V$ only contains elements in $P_i$ $(1\le i\le k$). Then with
probability at least $1-\gamma_4$, the approximate center $c_i=c(V)$
satisfies the inequality
\begin{eqnarray}
||c_i-m_i||^2\le \delta(\epsilon)\sigma_i^2,
\label{Ec1-closeness-ineqn}
\end{eqnarray}
\end{enumerate}
where $\gamma_4=\gamma^*$, which is defined in algorithm $k$-Means(.).
\end{lemma}

\begin{proof}Note that $\lim_{\epsilon\rightarrow +0}\delta(\epsilon)=0$.
    Let $\epsilon_0$ be a real number in $(0,1)$ such that
    \begin{eqnarray}
     e^{-{\delta_2\over 4\gamma^*\delta(\epsilon_0)}}\le {\mu_2\over 4}.\label{epsilon0p-ineqn}
    \end{eqnarray}

Statement~\ref{second-early1}:  Since $P_1|\ge |P_2|$, we have $|P_1|\ge {|P|\over k}$.
Let $p_1={1\over k}$. With $N_a$ elements from $P$, with
probability at most $\gamma_1=e^{{-\delta_2^2\over
2}p_1N_a}=e^{{-\delta_2^2d_2\over 2}M}$, there are less than
$(1-\delta_2)p_1N_a={(1-\delta_2)d_2}M$ elements from $P_1$ by
Theorem~\ref{chernoff3-theorem}. It is easy to see that
$\gamma_1=2^{{-\delta_2^2d_2\over 4}M}\le \gammaboundk$ for all
 $\epsilon\le \epsilon_0$ (by inequality~(\ref{epsilon0p-ineqn})).

Statement~\ref{second-early2}: Let $p={\delta(\epsilon)^{1+d_1}\over
k-1}$. By line~\ref{EN1b-line} of the algorithm, we have
$pN_b={M\over 1-\delta_2}$.  When $N_b$ elements are selected from
$P$, by Theorem~\ref{chernoff3-theorem}, with probability at most
$\gamma_2=2^{{-\delta_2^2\over 2}pN_b}=2^{{-\delta_2^2\over 2}\cdot
{M\over 1-\delta_2}}\le 2^{{-\delta_2^2M\over 4}}\le \gammaboundk$
for all $\epsilon\le \epsilon_0$  (by
inequality~(\ref{epsilon0p-ineqn})), we have less than
$(1-\delta_2)pN_b=M$ random points from $P_2$.

Statement~\ref{second-early3}: It takes ${N_a\choose M}\cdot
{N_b\choose M}$ cases to enumerate all possible subsets $H_1$ and $H_2$ of size $M$ from $S_a$ and $S_b$, respectively.

The number of pairs of centers
generated by line~\ref{Ephase1-start2-line} to line~\ref{Ephase1-end2-line} of the algorithm $k$-Means(.) is at most
\begin{eqnarray}
{N_a\choose M}\cdot {N_b\choose
    M}\label{T1-inequality12}
&\le&\left(\left({N_a^M\over M!}\right)\cdot
\left({N_b^M\over
    M!}\right)\cdot \right)\label{T1-inequality22}\\
&=&\bigO\left(\left({N_a\cdot e\over M}\right)^M\cdot
\left({N_b\cdot e\over
    M}\right)^M\right)\label{T1-inequality3}\\
&=&\bigO\left((ed_2k)^M\cdot \left({e(k-1)\over (1-\delta_2)\cdot
    \delta(\epsilon)^{1+d_1}}\right)^M\right)\label{T1-inequality42}\\
&=&\bigO\left( \left({d_2e^2k(k-1) \over (1-\delta_2)\cdot
\delta(\epsilon)^{1+d_1}}\right)^M\right)\label{T1-inequality52}\\
&=&\bigO\left( \left({(1+2\delta_2)\cdot d_2e^2k(k-1) \over
    \delta(\epsilon)^{1+d_1}}\right)^{M}\right).\label{T1-inequality72}
\end{eqnarray}

Statement~\ref{second-early4}: Let $\epsilon_1$ be a real number in
$(0,1)$ such that $\delta(\epsilon_1)^{1+d_1}M\le \gammaboundk$. Let
$z=\delta(\epsilon)^{1+d_1}M$. When $M $ elements are selected in
$P$, the probability that $V_a$ contains no element of $P_2$ is at
least $(1-z)^{M }\ge 1-zM\ge 1-\delta(\epsilon)^{1+d_1}M
=1-\delta(\epsilon)^{1+d_1}M=1-\delta(\epsilon)^{d_1}d_4$ by
Lemma~\ref{ex-lemma}, and equation (\ref{c4-eqn}).
Let $\gamma_3=z$. We have
$\gamma_3\le\gammaboundk$ for all  positive $\epsilon\le \epsilon_1$.

Statement~\ref{second-early5}: It follows from
Lemma~\ref{geometric-random-lemma}. 
\end{proof}

\begin{lemma}\label{failure-lemma}
With probability $1-(\mu_1+\mu_2+\mu_3)$, the algorithm outputs a list U of $k$-centers $(c_1,\cdots, c_k)$ such that one of them gives a $(1+\epsilon)$-approximation for the constrained $k$-means.
\end{lemma}

\begin{proof}We discuss two cases of the algorithm $k$-Means(.) based on the balanced conditions:

Case1:  ${\delta(\epsilon)^{1+d_1}\over
    k-1}|P|\le |P_2|\le |P_1|$.

This case is successful if 1) Multiset $S_a$ has a size $M$ multisubset $H_1$ with all elements are in $P_1$, 2) Multiset $S_b$ has a size $M$ multisubset $H_2$ with all elements in $P_2$,  3) $||c(H_1)-m_1||\le \delta(\epsilon)\sigma_1$, and
4) $||c(H_2)-m_2||\le \delta(\epsilon)\sigma_2$.

By Lemma~\ref{Eearly-sampling-lemma2},
the failure probability of this case
is at most
$\gamma_1+\gamma_2+\gamma_4+\gamma_4\le \mu_2+\mu_3$.

Case 2: $|P_2|<{\delta(\epsilon)^{1+d_1}\over k-1}\cdot |P|$.

This case is successful if 1) Multiset $S_a$ has a size $M$ multisubset $H_1$ with all elements are in $P_1$, and  2) $||c(H_1)-m_1||\le \delta(\epsilon)\sigma_1$.

By Lemma~\ref{Eearly-sampling-lemma2},
the failure probability of this case
is at most $\gamma_1+\gamma_4\le \mu_2+\mu_3$.

Algorithm $A(.)$ fails with
probability at most $\mu_1$.
Combining the two cases, we have that the entire failure
probability of the algorithm is at most $\mu_1+\mu_2+\mu_3<1$.
\end{proof}

\begin{definition}
    Let $\epsilon$ be a positive real number.   An algorithm $H(.)$ is a
{\it    $(1+\epsilon)$-solution} of the constrained $k$-means if it outputs a list $U$ of $k$ approximate centers $(c_1,\cdots, c_k)$ such that one $(c_1^*,\cdots, c_k^*)$ in  $U$ induces a $(1+\epsilon)$-approximation for the input of the constrained $k$-means problem.
\end{definition}

\begin{theorem}Assume that $A(k, P,\epsilon, C)$ is an $(\epsilon,\delta(.),\mu_1)$-\extendKM~with complexity $Z(k, n, d, \epsilon)$. Then there is a
    $\bigO(2^{\bigO({k\over \delta(\epsilon)})}+Z(k,k-1, n, d,
    \epsilon)+\left({e^2(d_2 +\delta)k^2\over
        \delta(\epsilon)^{1+d_1}}\right)^{M} Z(k,k-2, n, d,
    \epsilon))$ time algorithm to give a $(1+\epsilon)$-solution for the
constrained $k$-means. Furthermore, it outputs a list of $U$ of at most
$H(k,k-1, n,d,\epsilon)+\bigO\left( \left({e^2(d_2 +\delta)k^2\over
    \delta(\epsilon)^{1+d_1}}\right)^{M}\right)H(k,k-2, n,d,\epsilon)$ many $k$-centers, and its failure
probability is at most $\mu_1+\mu_2+\mu_3<1$, where the parameters
$\delta, d_2$ and $M$ are defined in algorithm $k$-means.
\end{theorem}

\begin{proof}Let $P_1\ge P_2\ge\cdots \ge P_k$ be an optimal
$k$-clusters for a  given constrained $k$-means problem. Assume that
$t$ is an integer with $1\le t\le 2$ such that we have the $t$
approximate centers $(c_1,\cdots, c_t)$ with $f_2(c_i, P_i)\le
(1+\delta(\epsilon))F_2(m_i, P_i)$ for $i=1,2,\cdots, t$. The
\extendKM~$A(.)$ will return the rest $k-t$ approximate centers with
satisfied accuracy. Let $H(M,N_a,N_b)={N_a\choose M}\cdot
{N_b\choose
    M}=\bigO\left( \left({(1+2\delta_2)\cdot d_2e^2k(k-1) \over
    \delta(\epsilon)^{1+d_1}}\right)^{M}\right)$ by Statement~\ref{second-early3} of Lemma~\ref{Eearly-sampling-lemma2}.

By Lemma~\ref{geometric-random-lemma}, a multiset of $M$ random
samples from $P_i$ can generate an approximate center $c_i$ to
$m_i$. It follows from Lemma~\ref{Eearly-sampling-lemma2}. As $A(.)$
finds the rest of centers, therefore, it gives a
$(1+\epsilon)$-approximation algorithm for the constrained $k$-means
problem by Definition~\ref{extend-k-means-def}.

The set $D$ has only one vector $(c_1)$ with a single approximate center
for $P_1$. It is added to $D$ at line~\ref{k-means-end-line} of the algorithm $k$-Means(.). Thus, it
introduces additional $H(k,k-1,n,d,\epsilon)$ items in the list.
The set $D$ has $H(M,N_a,N_b)$ vectors $(c_1, c_2)$ of
two centers, which are added to $D$ at line~\ref{add-two-center-line} of the algorithm $k$-Means(.). Thus, it introduces additional $H(M,N_a,N_b)$ items in the list.

Therefore, the  time complexity is $\bigO(2^{\bigO({k\over
\delta(\epsilon)})}+Z(k,k-1, n, d, \epsilon)+ H(M,N_a,N_b) Z(k,k-2,
n, d, \epsilon))$. It returns a list $U$ of at most
$H(k,k-1,n,d,\epsilon)+H(M,N_a,N_b)H(k,k-2,n,d,\epsilon)$
$k$-centers.

The failure probability follows from Lemma~\ref{failure-lemma}. It fails with probability at most $\mu_1+\mu_2$.
\end{proof}

\section{NP-Hardness for Balanced  $k$-means }
In this section, we show that the $1$-balanced $2$-means problem is
NP-hard. We derive a polynomial time reduction from the classical
balanced 2-partition problem. The reduction here is adapted to the method in~\cite{AloiseDHP09}, which shows that $2$-means problem is
NP-hard without balance restriction.

\subsection{NP-Completeness}

 Given a simple graph $G = (V;E)$,  a balanced 2-partition of $G =
(V,E)$ is a partition of $V$ into two vertex sets $V_1, V_2$ such that
$|V_i|\le \ceiling{|V|\over 2}$. The cut size (or simply, the size)
of a balanced $2$- partition is the number of edges of $G$ with one
endpoint in set $V_1$ and the other endpoint in  set $V_2$. The
maximum  bisection is to get a balanced $2$-partition with the
maximum cut size. This is a well known NP-hard problem~\cite{DiazMertzios17}.

\begin{proposition}
The maximum bisection problem with even number of vertices is
NP-hard.
\end{proposition}

\begin{proof}
We can derive a polynomial time reduction from the maximum bisection
problem $G=(V,E)$ to a maximum bisection problem $G'=(V',E')$. We
only consider the case that $|V|$ is odd. Let $u$ be a vertex not in
$V$. Let $V'=V\cup \{u\}$, and $E'=E$. It is easy to see the maximum
bisection cut size for $G'$ is identical to that of $G$.
\end{proof}

\begin{theorem}
The $1$-balanced $2$-means problem is NP-hard.
\end{theorem}

\begin{proof}
Let $G=(V,E)$ be an instance of the maximum bisection problem with even $|V|$. Let
$v_1, v_2,\cdots, v_n$ be the vertices in $V$, and $e_1,\cdots, e_m$
be the edges in $E$. Consider a $|V|\times |E|$ matrix $M$. Each
column corresponds to an edge. For each edge $e_k=(v_i, v_j)$ with
$i<j$, let the $i$-th item of $k$-th column be $1$, and $j$-th item
of $k$-column be $-1$. Let the constrained $2$-means  problem consist of the
points that correspond to the $|V|$ rows in the matrix $M$.

Let $P$ and $Q$ be a balanced $2$-partition of $V$. Let $C_P$ be the
center of $P$, and $C_Q$ be the center of $Q$. If an edge $e_k=(v_i,
v_j)$ is a crossing edge, then the $k$-th item in $C_P$ is either
${2\over n}$ or $-{2\over n}$, and  the $k$-th item in $C_Q$ is
either ${2\over n}$ or $-{2\over n}$. Otherwise, the $k$-th items in
$C_P$  and $C_Q$ are both zero.

Therefore, the square of distances to the two centers are
\begin{eqnarray*}
&&\sum_{e\in E(P,Q)}(({n\over 2}-1){4\over n^2}+(1-{2\over n})^2+({n\over 2}-1){4\over n^2}+(1-{2\over n})^2)+\sum_{E(P,P)} 2+\sum_{E(Q,Q)}2\\
&=&(2-{4\over n})|E(P,Q)|+2|E(P,P)|+2|E(Q,Q)|\\
&=&2|E|-{4\over n}\cdot |E(P,Q)|.
\end{eqnarray*}

Therefore, the maximum balanced bisection problem with even number
of vertices is polynomial time reducible to the balanced $2$-means
problem.
\end{proof}

\newpage

\centerline{\Large Appendix}

\newcommand{\vibrate}{{\rm vibrate}}

 For a point $q$ a set of points $P\subseteq \mathbb{R}^d$, our algorithm needs to find the $i$-th largest $\dist(q,p)$ for $p\in P$ (line~\ref{find-i-th-largest-line} of algorithm $2$-Means(.)). The following Lemma~\ref{vibrate-lemma} makes the point $p$ unique by a minor adjustment
for point $q$. When $q$ is an approximate center of the largest
cluster $P_1$ for a constrained $2$-means problem, $f_2(q', P_1)\approx f_2(q,P_1)$ if $\rho$ is small enough by Lemma~\ref{vibrate-propsition}.

\begin{definition}
    For a point $q\in \mathbb{R}^d$, a $\rho$-vibration of $q$ is to generate
    $q'=q+Y$, where $Y=(y_1,\cdots, y_n)$, and each $y_j$ is a random real
    number in $[-\rho,\rho]$ for $j=1,\cdots, n$. We use the notation
    $q'=\vibrate(q, \rho)$.
\end{definition}

Before reading Lemma~\ref{vibrate-lemma}, a reader may
feel comfortable about the simple fact that a random real number in
$[0,1]$ is not equal to $0.5$ with probability one.

\begin{lemma}\label{vibrate-lemma}
    Let $P=\{p_1,\cdots, p_n\}$ be a finite set of points in $\mathbb{R}^d$,
    $q\in \mathbb{R}^d$, and $\rho>0$. Let $q'$ be  a $\rho$-vibration from $q$.
    Then with probability $1$,
    $\dist(p_i, q')\not=\dist(p_j,q')$ for any $p_i\not=p_j$ in $P$.
\end{lemma}

\begin{proof}Assume that all points in $P$ are fixed, and $q$ is
    also fixed. Let $p_i=(a_1,\cdots, a_d)$ and $p_j=(b_1,\cdots, b_d)$
    and $q=(h_1,\cdots, h_d)$. For a point $y=(y_1,\cdots, y_d)\in \mathbb{R}^d$
    in a ball of radius $\rho$ with center at origin, we consider the case
    $||(q+y)-p_i||^2=||(q+y)-p_j||^2$. Without loss of generality, assume that
    $a_1\not=b_1$. We have $y_1=f(y_2,\cdots, y_d)$ for a quadratic
    polynomial function $f(.)$, which is nonzero since $p_i\not= p_j$.
    Therefore, the set of all points $\{y: y\in \mathbb{R}^d,{\rm \ and\ } ||(q+y)-p_i||^2-||(q+y)-p_j||^2=0\}$ is of dimensional
    less than $d$. Therefore, it is of measure $0$ in $\mathbb{R}^d$. There are
    at most $|P|^2$ pairs of $p_i$ and $p_j$. Therefore, with
    probability $1$,
    $\dist(p_i, q')\not=\dist(p_j,q')$ for any $p_i\not=p_j$ in $P$.
\end{proof}


\begin{lemma}\label{vibrate-propsition}There is a $\bigO(nd)$
    time algorithm such that given a parameter $\eta\in (0,1)$, a point $q\in \mathbb{R}^d$, and a set $P_1$ ($|P_1|\ge 3$) of $n$ points in $\mathbb{R}^d$ that
    are different each other, it generates  $q'=$vibrate$(q, \rho)$ with $f_2(q', P_1)\le
    (1+\eta)f_2(q, P_1)$, where $d_{\max}=\max\{\dist(q, p'):p'\in P\}$,
    $d_{min2}$ is the second smallest in $\{\dist(q,p'):  p'\in P\}$, and
    $\rho=\min\left(d_{max}, {\eta d_{min2}^2\over 3dd_{max}}\right)$.
\end{lemma}

\begin{proof}  For $q'=$vibrate$(q, \rho)$, we
    have $||q-q'||^2\le d\rho^2$. Thus, we have
    \begin{eqnarray}
    f_2(q',P_1)&=&\sum_{x\in P_1}||x-q'||^2\label{prop-ineqn0}\\
    &\le&\sum_{x\in P_1}(||x-q||+||q-q'||)^2\label{prop-ineqn1a}\\
    &\le&\sum_{x\in P_1}(||x-q||^2+2||x-q||\cdot ||q-q'||+||q-q'||^2)\label{prop-ineqn1}\\
    &\le&\sum_{x\in P_1}(||x-q||^2+2d\rho\cdot d_{max}+d\rho^2)\label{prop-ineqn2}\\
    &\le&\sum_{x\in P_1}(||x-q||^2+2d\rho\cdot d_{max}+d\rho\cdot d_{max})\\
    &\le&\sum_{x\in P_1}(||x-q||^2+3d\rho\cdot d_{max})\\
    &\le&\sum_{x\in P_1}(||x-q||^2+{d_{min2}^2\eta})\\
    &\le&\left(\sum_{x\in P_1}||x-q||^2\right)+{f_2(q, P_1)\eta}\\
    &\le& (1+\eta)f_2(q, P_1).
    \end{eqnarray}
    The transition from (\ref{prop-ineqn0}) to (\ref{prop-ineqn1a}) is by
    triangular inequality.
\end{proof}

\end{document}